\documentclass[a4paper,11pt]{amsart}
\usepackage[letterpaper, margin=1in]{geometry}
\usepackage{latexsym}
\usepackage{amssymb}
\usepackage{amsmath}
\usepackage{amsfonts}
\usepackage{color}

            \DeclareFontFamily{U}{wncy}{}
            \DeclareFontShape{U}{wncy}{m}{n}{%
               <5>wncyr5%
               <6>wncyr6%
               <7>wncyr7%
               <8>wncyr8%
               <9>wncyr9%
               <10>wncyr10%
               <11>wncyr10%
               <12>wncyr6%
               <14>wncyr7%
               <17>wncyr8%
               <20>wncyr10%
               <25>wncyr10}{}

\newtheorem{thm}{Theorem}[section]

\newtheorem{lem}[thm]{Lemma}

\newtheorem{cor}[thm]{Corollary}
\newtheorem{prop}[thm]{Proposition}

\theoremstyle{definition}
\newtheorem*{dfn}{Definition}
\newtheorem{remark}{Remark}[section]

\newtheorem{example}{Example}[]

\newcommand{\N}{\mathbb N}
\newcommand{\Z}{\mathbb Z}
\newcommand{\F}{\mathbb F}

\def\C{{\mathcal C}}

\def\ga{\gamma}
\def\s{\sigma}
\def\d{\delta}

\def\la{\lambda}

\def\ol{\overline}

\begin{document}

\title[Hulls of Matrix-Product Codes]{The Hulls of Matrix-Product Codes over Commutative Rings and Applications}

\author{Abdulaziz Deajim}
\address[A. Deajim]{Department of Mathematics, King Khalid University,
P.O. Box 9004, Abha, Saudi Arabia} \email{deajim@kku.edu.sa, deajim@gmail.com}

\author{Mohamed Bouye}
\address[M. Bouye]{Department of Mathematics, King Khalid University,
P.O. Box 9004, Abha, Saudi Arabia} \email{medeni.doc@gmail.com}

\author{Kenza Guenda}
\address[K. Guenda]{Faculty of Mathematics, USTHB, Laboratory of Algebra and Number Theory, BP 32  El Alia, Bab Ezzouar, Algeria} \email{ken.guenda@gmail.com}

\date{\today}

\begin {abstract}
Given a commutative ring $R$ with identity, a matrix $A\in M_{s\times l}(R)$, and $R$-linear codes $\mathcal{C}_1, \dots, \mathcal{C}_s$ of the same length, this article considers the hull of the matrix-product codes $[\mathcal{C}_1 \dots \mathcal{C}_s]\,A$. Consequently, it introduces various sufficient conditions under which $[\mathcal{C}_1 \dots \mathcal{C}_s]\,A$ is a linear complementary dual (LCD) code. As an application, LCD matrix-product codes arising from torsion codes over finite chain rings are considered. Highlighting examples are also given.
\end {abstract}
\maketitle

\section{{\bf Introduction}}\label{intro}
An active theme of research in coding theory is the construction of new codes by modifying or combining existing codes. In 2001, Blackmore and Norton \cite{BN} introduced the interesting and useful construction of matrix-product codes over finite fields. Such a construction included as special cases some previously well-known constructions such as the Plotkin's $(u|u+v)$-construction, the $(u+v+w|2u+v|u)$-construction, the Turyn's $(a+x|b+x|a+b+x)$-construction, and the $(u+v|u-v)$-construction. Through subsequent efforts of many researchers, matrix-product codes were further studied over finite fields and some types of finite commutative rings, see for instance \cite{As}, \cite{BD}, \cite{CLL}, \cite{FLL}, \cite{GHR}, \cite{HLR}, \cite{HR}, and \cite{MS}.

In \cite{Mas}, J.L. Massey introduced the notion of linear complementary dual (LCD) codes over finite fields. Ever since, many subsequent papers on LCD codes and their applications over finite fields and some finite commutative rings have appeared, see for instance \cite{DKOSS}, \cite{GOS}, \cite{JSU}, \cite{LDL}, \cite{LL1}, \cite{MGGSS}, and \cite{TJ}.

In a follow-up to \cite{DB}, we consider some aspects that connect the above tow notions: matrix-product codes and LCD codes over commutative rings. For this purpose, we first focus on studying the hull of matrix-product codes over such rings. We then use this to introduce various sufficient conditions under which a matrix-product code is an LCD code. As an application, LCD matrix-product codes arising from torsion codes over finite chain rings are considered. Highlighting examples are also given.

In order to put our results in a context as broad as possible, we assume in this article, unless otherwise stated, that {\bf $R$ stands for a commutative ring with identity 1}. We denote by $U(R)$ the multiplicative group of units of $R$. For the sake of completeness, we set below relevant terminologies and remind the reader of some facts needed in this article.

\subsection{Matrices}$\\$
\indent For positive integers $s$ and $l$, denote by $M_{s\times l}(R)$ the set of $s\times l$ matrices over $R$. In this article, we always assume that $s\leq l$. A square matrix over $R$ is called non-singular if its determinant is in $U(R)$. With $A\in M_{s\times l}(R)$, $A$ is said to be of {\it full row rank} (FRR) if its $s$ rows are linearly independent over $R$ (see \cite[Definition 2.4]{FLL}). Denoting the $s\times s$ identity matrix by $I_s$, $A$ is said to be {\it right-invertible} if there is a matrix $B\in M_{l\times s}(R)$, called a {\it right inverse} of $A$, such that $AB=I_s$. Left-invertibility of $A$ is defined similarly. If further $R$ is finite, then $A$ is right-invertible if and only if it is FRR (\cite[Corollary 2.7]{FLL}). A square matrix over $R$ is right-invertible if and only if it is left-invertible, in which case left and right inverses are equal and the square matrix is said to be {\it invertible} (\cite[p. 10]{Mc}). If further $R$ is finite, then a square matrix over $R$ is non-singular if and only if it is FRR if and only if it is invertible (\cite[Corollary 2.8]{FLL}).

If $R$ is a finite field and $A\in M_{s\times l}(R)$, we say that $A$ is {\it non-singular by columns} if, for every $1\leq t \leq s$, every $t\times t$ submatrix of $A_t$ is non-singular, where $A_t$ is the submatrix of $A$ consisting of the upper $t$ rows. This notion was first introduced over finite fields in \cite{BN} and then extended in \cite{FLL} to finite commutative Frobenius rings (and hence to finite commutative chain rings which we shall need in Section 3). 

By $\mbox{diag}(r_1, \dots, r_s)\in M_{s\times s}(R)$, we mean an $s\times s$ diagonal matrix whose diagonal entry in position $(i,i)$ is $r_i\in R$ for $i=1, \dots, s$; while by $\mbox{adiag}(r_1, \dots, r_s)\in M_{s\times s}(R)$ we mean an $s\times s$ anti-diagonal matrix whose anti-diagonal entry in position $(i, s-i+1)$ is $r_i\in R$ for $i=1, \dots, s$. If $A\in M_{s\times s}(R)$ is such that $AA^t=\mbox{diag}(r_1, \dots, r_s)$ or $AA^t=\mbox{adiag}(r_1, \dots, r_s)$ with $r_i\in U(R)$ for $i=1, \dots, s$ and $A^t$ denotes the usual transpose of $A$, then both $A$ and $A^t$ are non-singular. This is because the relevant properties of determinants over fields remain valid over commutative rings (see \cite{Mc})).

\subsection{Matrix-Product Codes}$\\$
\indent A {\it code} $\C$ over $R$ of length $m\in \N$ is simply a subset of $R^m$. If further $\C$ is an $R$-submodule of $R^m$, then we call it an {\it $R$-linear code} or just a {\it linear code}. The distance on codes is meant here to be the Hamming distance, and we denote the minimum distance of $\C$ by $d(\C)$. We say that a linear code $\C$ is {\it free of rank} $k$ over $R$ if it is so as an $R$-module. If, in addition, $\C$ is of minimum distance $d$, then we say that it is an $[m,k,d]$-linear code over $R$.

Given a matrix $A\in M_{s\times l}(R)$ and codes $\C_1, \dots, \C_s$ of length $m$ over $R$, define the {\it matrix-product code} $[\C_1 \dots \C_s]\,A$ to be the code over $R$ whose codewords are the $m\times l$ matrices $(c_1 \dots c_s)\,A$, where $c_i\in \C_i$ (as a column-vector) for $i=1, \dots, s$ (see \cite{As} or \cite{BN} for instance). In particular, we denote $[\C_1 \dots C_s]\,I_s$ by $[\C_1 \dots \C_s]$. It should be noted that the matrix-product code $[\C_1 \dots \C_s]\,A$ can also be thought of as a code of length $ml$ over $R$ in an obvious way. The codes $\C_1, \dots, \C_s$ are called the {\it input codes} of the matrix-product code $[\C_1 \dots \C_s]\,A$. It can be easily checked that a matrix-product code is linear over $R$ so longs as all of its input codes are linear over $R$, and this will be our assumption on the input codes throughout the article.

\subsection{Hulls and LCD Codes}$\\$
\indent Consider the (Euclidean) bilinear form $$\langle .\, ,.\rangle: R^m \times R^m \to R$$ defined by $$\langle (x_1, \dots, x_m), (y_1, \dots, y_m)\rangle = \sum_{i=1}^m x_i y_i.$$
Let $\C$ be a linear code over $R$ of length $m$. Define the (Euclidean) {\it dual} $\C^\perp$ of $\C$ to be the following linear code:
$$\C^\perp =\{ y\in R^m\,|\, \mbox{$\langle x, y \rangle =0$ for all $x\in \C$}\};$$
that is, $\C^\perp$ consists of the elements of $R^m$ which are orthogonal to all elements of $\C$ with respect to the Euclidean bilinear form. The {\it hull} $H(\C)$ of $\C$ is the linear code $\C \cap \C^\perp$. We call $\C$ a {\it linear complementary dual} (LCD) code over $R$ if its hull is trivial (i.e. $H(\C)=\{0\}$).

\subsection{Contributions of the Article}$\\$
\indent Motivated by the ongoing developments of matrix-product codes and the numerous applications of the hulls of linear codes, we study in Section 2 the hull of a matrix-product code over $R$ and show how it is related to the hulls of its input codes. We then utilize our results to give constructions of LCD matrix-product codes over $R$. Some of our results generalize their counterparts over finite fields which appeared in \cite{LL2}. As an application, we visit in Section 3 the concept of torsion codes over finite chain rings. We use such codes along with results from Section 2 to further construct LCD matrix-product codes over the residue fields of the underlying rings and consider their minimum distances. Several examples are given throughout the article.

\section{{\bf The Hulls of Matrix-Product Codes}}\label{LCD MPC}


For a non-singular $A\in M_{s\times s}(R)$ and free linear codes $\C_1, \dots, \C_s$ over $R$ of the same length, the equality $$([\C_1 \dots \C_s]\,A)^\perp=[\C_1^\perp \dots \C_s^\perp]\,(A^{-1})^t$$ was shown to hold when $R$ is a finite field (\cite{BN}), a finite chain ring (\cite{As}), a finite commutative ring (\cite{BD}), and, most recently, an arbitrary commutative ring without even requiring the freeness of the input codes (\cite{DB}).

\begin{prop}\label{MPC dual} $($\cite[Theorem 3.3]{DB}$)$
Let $\C_1, \dots, \C_s$ be linear codes over $R$ of the same length and $A\in M_{s\times s}(R)$ non-singular. Then, $([\C_1 \dots \C_s]\,A)^\perp=[\C_1^\perp \dots \C_s^\perp]\,(A^{-1})^t$.
\end{prop}

\begin{lem}\label{2 conditions}
Let $\C_1, \dots, \C_s$ be linear codes over $R$ of the same length and $A\in M_{s\times s}(R)$. If either of the following holds:
\begin{itemize}
\item[1.] $AA^t=\mbox{diag}(r_1, \dots, r_s)$ for $r_1, \dots, r_s\in U(R)$, or
\item[2.] $AA^t=\mbox{adiag}(r_1, \dots, r_s)$ for $r_1, \dots, r_s \in U(R)$, and $\C_i^\perp=\C_{s-i+1}^\perp$ for $i=1, \dots, s$,
\end{itemize}
then $([\C_1 \dots \C_s]\,A)^\perp=[\C_1^\perp \dots \C_s^\perp]\,A$.
\end{lem}

\begin{proof}\hfill
\begin{itemize}
\item[1.] Assume that $AA^t=\mbox{diag}(r_1, \dots, r_s)$ for $r_1, \dots, r_s\in U(R)$. Then $A$ is non-singular, \linebreak $(A^{-1})^t=\mbox{diag}(r_1^{-1}, \dots,r_s^{-1})A$ and, by Proposition \ref{MPC dual},
    \begin{align*}
    ([\C_1 \dots \C_s]\,A)^\perp &=[\C_1^\perp \dots \C_s^\perp]\,\mbox{diag}(r_1^{-1}, \dots,r_s^{-1})A \\
    &=[r_1^{-1}\C_1^\perp \dots r_s^{-1}\C_s^\perp]\,A. 
    \end{align*}
    Since the codes $\C_i$ are linear over $R$ and $r_i^{-1}$ are units in $R$, $r_i^{-1}\C_i^\perp=\C_i^\perp$ for $i=1, \dots, s$ and, thus, the desired conclusion follows.
\item[2.] Assume that $AA^t=\mbox{adiag}(r_1, \dots, r_s)$ for $r_1, \dots, r_s \in U(R)$, and $\C_i^\perp=\C_{s-i+1}^\perp$ for $i=1, \dots, s$. Then $A$ is non-singular, $(A^{-1})^t=\mbox{adiag}(r_s^{-1}, \dots, r_1^{-1})A$ and, by Proposition \ref{MPC dual},
\begin{align*}
([\C_1 \dots \C_s]\,A)^\perp &= [\C_1^\perp \dots \C_s^\perp]\,\mbox{adiag}(r_s^{-1}, \dots,r_1^{-1})A \\ 
&= [r_1^{-1}\C_s^\perp \dots r_s^{-1}\C_1^\perp]\,A \\
&= [r_1^{-1}\C_1^\perp \dots r_s^{-1}\C_s^\perp]\,A.
\end{align*}
Since the codes $\C_i$ are linear over $R$ and $r_i^{-1}$ are units in $R$, $r_i^{-1}\C_i^\perp=\C_i^\perp$ for $i=1, \dots, s$ and, thus, the desired conclusion follows.
\end{itemize}
\end{proof}

\begin{lem}\label{4 conditions} $($\cite[Lemma 3.6]{DB}$)$
Let $\C_1, \dots, \C_s$ be linear codes over $R$ of the same length and \linebreak $A\in M_{s\times s}(R)$ non-singular. If either of the following holds:
\begin{itemize}
\item[1.] $\C_1 \subseteq \C_2 \subseteq \dots \subseteq \C_s$ and $A$ is upper triangular,
\item[2.] $\C_s \subseteq \C_{s-1} \subseteq \dots \subseteq \C_1$ and $A$ is lower triangular,
\item[3.] $A$ is diagonal, or
\item[4.] $\C_1 = \C_2 = \dots = \C_s$,
\end{itemize}
then $[\C_1 \dots \C_s]\,A=[\C_1 \dots \C_s]$.
\end{lem}

The following main result and the subsequent corollaries present conditions under which the hull of a matrix-product code is given in terms of the hulls of its input codes.

\begin{thm}\label{LCD thm}
Let $\C_1, \dots, \C_s$ be linear codes over $R$ of the same length \textcolor[rgb]{0.00,0.07,1.00}{and $A\in M_{s\times l}(R)$.}
\begin{itemize}
\item[(1)] If $([\C_1 \dots \C_s]\,A)^\perp=[\C_1^\perp \dots \C_s^\perp]\,A$ and $A$ is either FRR or right-invertible, then $$H([\C_1 \dots \C_s]\,A)=[H(\C_1) \dots H(\C_s)]\,A.$$
\item[(2)] If $[\C_1 \dots \C_s]\,A=[\C_1 \dots \C_s]$, then $$H([\C_1 \dots \C_s]\,A)=[H(\C_1) \dots H(\C_s)].$$
\end{itemize}
\end{thm}

\begin{proof}\hfill
\begin{itemize}
\item[(1)] Assume that $([\C_1 \dots \C_s]\,A)^\perp=[\C_1^\perp \dots \C_s^\perp]\,A$. Let $x \in H([\C_1 \dots \C_s]\,A)$. Then $x\in [\C_1 \dots \C_s]\,A$ and $x\in [\C_1^\perp \dots \C_s^\perp]\,A$. So, $x=(c_1 \dots c_s)\,A=(c_1' \dots c_s')\,A$ for some $c_i\in \C_i$, $c_i'\in \C_i^\perp$, $i=1, \dots, s$. If $A$ is right-invertible or the $s$ rows of $A$ are linearly independent over $R$, then we obviously get $c_i=c_i'\in H(\C_i)$ for every $i=1, \dots, s$, and so $x\in [H(\C_1) \dots H(\C_s)]\,A$. Thus, $H([\C_1 \dots \C_s]\,A) \subseteq [H(\C_1) \dots H(\C_s)]\,A$. Conversely, let $y\in [H(\C_1) \dots H(\C_s)]\,A$. So, $y=(y_1, \dots y_s)\,A$ for some $y_i\in H(\C_i)$, $i=1, \dots, s$. It follows that $y\in [\C_1 \dots, \C_s]\,A \cap [\C_1^\perp \dots \C_s^\perp]\,A=H([\C_1 \dots \C_s]\,A)$. Thus, $[H(\C_1) \dots H(\C_s)]\,A \subseteq H([\C_1 \dots \C_s]\,A)$.
\item[(2)] Assume that $[\C_1 \dots \C_s]\,A=[\C_1 \dots \C_s]$. It follows from Proposition \ref{MPC dual} that 
    $$([\C_1 \dots \C_s]\,A)^\perp=([\C_1 \dots \C_s])^\perp= ([\C_1 \dots \C_s]\,I_s)^\perp=[\C_1^\perp \dots \C_s^\perp]\, (I_s^{-1})^t=[\C_1^\perp \dots \C_s^\perp].$$ Thus, $H([\C_1 \dots \C_s]\,A)=[\C_1 \dots \C_s] \cap [\C_1^\perp \dots \C_s^\perp]=[H(\C_1) \dots H(\C_s)]$.
\end{itemize}
\end{proof}

\begin{remark} In Theorem \ref{LCD thm}, if $R$ is finite then the two properties FRR and right-invertibility of $A$ are equivalent (see \cite[Corollary 2.7]{FLL}). Furthermore, if $A$ is square, then the properties FRR, right-invertibility, invertibility, and non-singularity are all equivalent properties of $A$ (see \cite[Corollary 2.8]{FLL}).
\end{remark}

\begin{cor}\label{hull 1}
Let $\C_1, \dots, \C_s$ be linear codes over $R$ of the same length and $A\in M_{s\times s}(R)$ non-singular.
\begin{itemize}
\item[1.] If any of the two conditions of Lemma \ref{2 conditions} holds, then $H([\C_1 \dots \C_s]\,A)=[H(\C_1) \dots H(\C_s)]\,A$. 
\item[2.] If any of the four conditions of Lemma \ref{4 conditions} holds, then $H([\C_1 \dots \C_s]\,A)=[H(\C_1) \dots H(\C_s)]$.
\end{itemize}
\end{cor}

\begin{proof}
This is a direct application of Theorem \ref{LCD thm} along with Lemmas \ref{2 conditions} and \ref{4 conditions}.
\end{proof}

As a consequence of Theorem \ref{LCD thm} and Corollary \ref{hull 1}, the following result gives several sufficient conditions to characterize LCD matrix-product codes over commutative rings in terms of properties of their input codes and the matrix used.

\begin{cor}\label{LCD cor 1}
Let $\C_1, \dots, \C_s$ be linear codes over $R$ of the same length. Suppose that one of the following holds:
\begin{itemize}
\item[1.] $A\in M_{s\times l}(R)$, $([\C_1 \dots \C_s]\,A)^\perp=[\C_1^\perp \dots \C_s^\perp]\,A$, and $A$ is either FRR or right-invertible.
\item[2.] $A\in M_{s\times l}(R)$ and $[\C_1 \dots \C_s]\,A=[\C_1 \dots \C_s]$.
\item[3.] $A\in M_{s\times s}(R)$ and $AA^t=\mbox{diag}(r_1, \dots, r_s)$ with $r_i\in U(R)$, $i=1, \dots, s$.
\item[4.] $A\in M_{s\times l}(R)$, $AA^t=\mbox{adiag}(r_1, \dots, r_s)$ with $r_i \in U(R)$ and $\C_i^\perp=\C_{s-i+1}^\perp$, $i=1, \dots, s$.
\item[5.] $A\in M_{s\times s}(R)$ is non-singular upper triangular and $\C_1 \subseteq \C_2 \subseteq \dots \subseteq \C_s$.
\item[6.] $A\in M_{s\times s}(R)$ is non-singular lower triangular and $\C_s \subseteq \C_{s-1} \subseteq \dots \subseteq \C_1$.
\item[7.] $A\in M_{s\times s}(R)$ is non-singular and $\C_1 = \C_2 = \dots = \C_s$.
\end{itemize}
Then $[\C_1 \dots \C_s]\,A$ is LCD if and only if $\C_i$ is LCD for every $i=1, \dots, s$.
\end{cor}

\begin{proof}
It follows from Theorem \ref{LCD thm} and Corollary \ref{hull 1} that if any of the conditions 1, 3, or 4 occurs, then $H([\C_1 \dots \C_s]\,A)=[H(\C_1) \dots H(\C_s)]\,A$. In these cases, since $A$ is either FRR or right-invertible (in condition 1) and non-singular (in conditions 3 and 4), it follows that $H([\C_1 \dots \C_s]\,A)=\{0\}$ if and only if $H(\C_i)=\{0\}$ for every $i=1, \dots, s$ as desired. As for the remaining cases, we have, by Theorem \ref{LCD thm} and Corollary \ref{hull 1} again, the equality $H([\C_1 \dots \C_s]\,A)=[H(\C_1) \dots H(\C_s)]$, from which the desired conclusion is obvious.
\end{proof}

\begin{remark}
Condition 3 (resp. condition 6) of Corollary \ref{LCD cor 1} generalizes \cite[Theorem 3.1]{LL2} (resp. \cite[Theorem 3.2]{LL2}).
\end{remark}

\begin{example}
Let $\C_1=15 \Z_{30} \times 15\Z_{30}$, and $\C_2=10\Z_{30}\times 10\Z_{30}$. So, $\C_1^\perp=2\Z_{30}\times 2\Z_{30}$ and $\C_2^\perp=3\Z_{30} \times 3\Z_{30}$. It is then clear that both $\C_1$ and $\C_2$ are LCD codes over $\Z_{30}$. Let $A=\left(\begin{array}{cc}6&5\\5&6\end{array}\right)\in M_{2\times 2}(\Z_{30})$. Since $AA^t= \left(\begin{array}{cc}1&0\\0&1\end{array}\right)$, it follows from Corollary \ref{LCD cor 1} that $[\C_1\; \C_1]\,A$, $[\C_2 \; \C_2]\,A$, $[\C_1 \; \C_2]\,A$, and $[\C_2 \; \C_1]\,A$ are all LCD codes. For the sake of illustration, let us consider $[\C_1 \; \C_2]\,A$. As \linebreak $[\C_1 \; \C_2]=\left(\begin{array}{cc}15\Z_{30}&10\Z_{30}\\15\Z_{30}&10\Z_{30}\end{array}\right)$, it can be easily checked that
$$[\C_1 \C_2]\,A=\left(\begin{array}{cc}15\Z_{30}&10\Z_{30}\\15\Z_{30}&10\Z_{30}\end{array}\right)\,
\left(\begin{array}{cc}6&5\\5&6\end{array}\right)=
\left(\begin{array}{cc}10\Z_{30}&15\Z_{30}\\10\Z_{30}&15\Z_{30}\end{array}\right)=[\C_2 \C_1]=[\C_2 \C_1]\left(\begin{array}{cc}1&0\\0&1\end{array}\right).$$
Since $\left(\begin{array}{cc}1&0\\0&1\end{array}\right)$ is diagonal, it follows from Corollary \ref{LCD cor 1} that $[\C_2 \C_1]\left(\begin{array}{cc}1&0\\0&1\end{array}\right)$, and hence $[\C_1 \C_2]\,A$, is LCD.
\end{example}

\begin{example}
Let $u\in \Z_{25}$ be such that $u^2=-1$ (e.g. $u=7$). Then, the matrix $A=\left(\begin{array}{cc} 1&u\\u&1 \end{array}\right)$ is non-singular and satisfies $AA^t=\mbox{adiag}(2u,2u)$. In fact, $A$ is non-singular by columns (see \cite{BN} for more on this notion). It can be checked that $x^{12}-1$ factors into irreducible factors over $\Z_{25}$ as follows:
$$x^{12}-1=(x+1)(x-1)(x+7)(x-7)(x^2+x+1)(x^2+7x-1)(x^2-7x-1)(x^2-x+1).$$
Obviously, each irreducible factor of $x^{12}-1$ generates a cyclic code of length 12 over $\Z_{25}$. Using the LCD test (namely, $GG^t\in U(\Z_{25})$ where $G$ is a generating matrix of the code, see \cite[Theorem 3.5]{LL1}), it can be shown that, out of these cyclic codes, only $\C_1=\langle x+1 \rangle$ , $\C_2=\langle x^2+x+1 \rangle$, and $\C_3=\langle x^2-x+1 \rangle$ are LCD codes. Moreover, the parameters of $\C_1$, $\C_2$, and $\C_3$ are $[12, 11, 2]$, $[12, 10, 2]$, and $[12, 10, 2]$ respectively. It then follows from Corollary \ref{LCD cor 1} that $[\C_1 \; \C_1]\,A$, $[\C_2 \; \C_2]\,A$, and $[\C_3 \; \C_3]\,A$ are LCD codes. Their parameters are, respectively, $[24, 22, 2]$, $[24, 20, 2]$, and $[24, 20, 2]$ (see \cite{BD} for a bound on the minimum distance of a matrix-product code over any commutative ring). Note that in each of the obtained LCD codes, the length is doubled, the number of codewords is increased, and the information rate $k/n$ and minimum distance are maintained.
\end{example}

\begin{dfn} (\cite{FLL})
Let $A\in M_{s\times l}(R)$, $1\leq s_1, s_2 \leq s-1$ with $s_1 +s_2=s$, $A_1\in M_{s_1\times l}(R)$ is the matrix whose rows ar the upper $s_1$ rows of $A$, and $A_2\in M_{s_2\times l}(R)$ is the matrix whose rows are the lower $s_2$ rows of $A$. We say that $A$ has the {\it $s_1$-partitioned orthogonal property} if every row of $A_1$ is orthogonal to every row of $A_2$. With $A$ as such, we write $A=\left( \begin{array}{c} A_1 \\ \hline  A_2 \end{array}\right)$ and call the ordered pair $(A_1, A_2)$ the {\it $s_1$-orthogonal-property blocks} of $A$.
\end{dfn}

\begin{thm}\label{orth-property thm}
Let $\C_1$ and $\C_2$ be two linear codes of the same length over $R$, and assume that $A\in M_{s\times s}(R)$ is non-singular and has the $s_1$-partitioned orthogonal property. Then,
$$H([\underbrace{\C_1 \dots \C_1}_{s_1} \underbrace{\C_2 \dots \C_2}_{s_2}]\,A) \subseteq [\underbrace{H(\C_1) \dots H(\C_1)}_{s_1} \underbrace{H(\C_2) \dots H(\C_2)}_{s_2}]\,(A^{-1})^t.$$

\end{thm}

\begin{proof}
Let $(A_1, A_2)$ be the $s_1$-orthogonal-property blocks of $A$. So we have
$$AA^t=\left( \begin{array}{c} A_1 \\ \hline  A_2 \end{array}\right)\, \left(\begin{array}{c} A_1^t | A_2^t \end{array}\right)=\left(\begin{array}{cc} A_1 A_1^t & A_1 A_2^t \\ A_2 A_1^t & A_2 A_2^t \end{array}\right)=\left(\begin{array}{cc} A_1 A_1^t & 0\\ 0 & A_2 A_2^t\end{array}\right).$$
Let $z\in H([\underbrace{\C_1 \dots \C_1}_{s_1} \underbrace{\C_2 \dots \C_2}_{s_2}]\,A)$. Then, by Proposition \ref{MPC dual}, $$z=(x_1 \dots x_{s_1} y_1 \dots y_{s_2})\,A=(x_1' \dots x_{s_1}' y_1' \dots y_{s_2}')\,(A^{-1})^t,$$ for some $x_i\in \C_1, y_j\in \C_2, x_i'\in \C_1^\perp, y_j'\in \C_2^\perp$, $i=1, \dots, s_1$, $j=1, \dots, s_2$. It now follows that
$$(x_1 \dots x_{s_1} y_1 \dots y_{s_2})\,AA^t=(x_1' \dots x_{s_1}' y_1' \dots y_{s_2}').$$
Set $(x_1 \dots x_{s_1})\,A_1 A_1^t = (a_1 \dots a_{s_1})$ and $(y_1 \dots y_{s_2})\, A_2 A_2^t=(b_1 \dots b_{s_2})$. Since $\C_1$ and $\C_2$ are linear, we get $a_i\in \C_1$ and $b_j\in \C_2$ for $i=1, \dots, s_1$, $j=1, \dots, s_2$. So, we have $$(a_1 \dots a_{s_1} b_1 \dots b_{s_2})= (x_1' \dots x_{s_1}' y_1' \dots y_{s_2}').$$ Thus, $x_i' = a_i \in H(\C_1)$ for $i=1, \dots, s_1$ and $y_j'=b_j \in H(\C_2)$ for $j=1, \dots, s_2$. Hence, $$z\in [\underbrace{H(\C_1) \dots H(\C_1)}_{s_1} \underbrace{H(\C_2) \dots H(\C_2)}_{s_2}]\,(A^{-1})^t$$ and, therefore, $H([\underbrace{\C_1 \dots \C_1}_{s_1} \underbrace{\C_2 \dots \C_2}_{s_2}]\,A) \subseteq [\underbrace{H(\C_1) \dots H(\C_1)}_{s_1} \underbrace{H(\C_2) \dots H(\C_2)}_{s_2}]\,(A^{-1})^t$.
\end{proof}

\begin{cor}\label{orth-property cor}
Keep the assumptions of Theorem \ref{orth-property thm}. If $\C_1$ and $\C_2$ are both LCD codes, then so is $[\underbrace{\C_1 \dots \C_1}_{s_1} \underbrace{\C_2 \dots \C_2}_{s_2}]\,A$.
\end{cor}

\begin{proof}
A direct application of Theorem \ref{orth-property thm}.
\end{proof}

\begin{example}
Assume that $R$ is of characteristic 2. It is clear that $A=\left(\begin{array}{ccc} 1 & 0 & 1\\ 0 & 1 & 1\\ 1 & 1 & 1 \end{array}\right)$ is non-singular and has the 2-partitioned orthogonal property, with $A_1=\left(\begin{array}{ccc} 1 & 0 & 1\\ 0 & 1 & 1 \end{array} \right)$ and $A_2 = \left(\begin{array}{ccc} 1 & 1 & 1 \end{array}\right)$. By Corollary \ref{orth-property cor}, if $\C_1$ and $\C_2$ are LCD codes over $R$, then so are the matrix-product codes $[\C_1\, \C_1\, \C_2]\, A$ and $[\C_2\, \C_2\, \C_1]\,A$. Note that these two matrix-product codes correspond to the Turyn's $(a+x\,|\, b+x\,|\, a+b+x)$-construction.
\end{example}

\begin{remark}
Theorem \ref{orth-property thm} and Corollary \ref{orth-property cor} can be easily generalized to deal with more than two codes and a matrix with more than two submatrices concatenated vertically, where every row of a submatrix is orthogonal to every row of the other submatrices.
\end{remark}

\section{{\bf Matrix-Product Codes Arising from Torsion Codes}}

We utilize here some results from Section 2 to give constructions of LCD matrix-product codes over the residue field of a finite chain ring with torsion codes as input codes.

Recall that a finite commutative ring is a chain ring if it is a local principal ideal ring (see \cite[p. 339 and beyond]{Mc1} for more). In this section, we assume that {\bf $R$ is a finite commutative chain ring}. Let $\langle \ga \rangle$ be the maximal ideal of $R$ with $e$ the nilpotency index of $\ga$ and $R/\langle \ga \rangle$ the (finite) residue field of $R$. It follows that we have the following chain of ideals in $R$
$$\{0\} = \langle \ga^e \rangle \subseteq \langle \ga^{e-1} \rangle \subseteq \dots \subseteq \langle \ga \rangle \subseteq R.$$

For $r\in R$, denote by $\ol{r}\in R/\langle \ga \rangle$ the reduction of $r$ modulo $\langle \ga \rangle$. For $x=(x_1, \dots, x_m)\in R^m$, a code $\C$ of length $m$ over $R$, and $A=(a_{i,j})\in M_{s\times l}(R)$, set $\ol{x}=(\ol{x_1}, \dots, \ol{x_m})\in (R/\langle \ga \rangle)^m$, $\ol{\C}=\{\ol{x}\,|\, x\in \C\}$, and $\ol{A}$ the matrix $(\ol{a_{i,j}})\in M_{s\times l}(R/\langle \ga \rangle)$. The code $\ol{\C}$ (resp. the matrix $\ol{A}$) is called {\it the reduction code} of $\C$ (resp. {\it the reduction matrix} of $A$) modulo $\langle \ga \rangle$.

\begin{lem}\label{units}
For $u\in R$, $u\in U(R)$ if and only if $\ol{u}\neq \ol{0}$ in $R/\langle \ga \rangle$. 
\end{lem}
\begin{proof}
This is clear when taking into account that $U(R)=R-\langle \ga \rangle$ since $R$ is local and commutative.
\end{proof}

The following result shows that an NSC matrix over $R$ can be gotten from an NSC matrix over $R/\langle \ga \rangle$ and vice versa.
\begin{lem}\label{NSC} A matrix $A\in M_{s\times l}(R)$ is NSC if and only if $\ol{A}$ is NSC.
\end{lem}
\begin{proof}
An immediate consequence of Lemma \ref{units} is that a square matrix over $R$ is invertible if and only if its reduction modulo $\langle \ga \rangle$ is invertible (see also \cite[p. 177]{BC}). Hence, for $t=1, \dots, s$, a $t\times t$ submatrix of $A_t$ is non-singular if and only if its reduction modulo $\langle \ga \rangle$ is non-singular. Since $R$ is finite, the proof follows from the fact that non-singularity and invertibility are two equivalent properties of a square matrix over $R$ (\cite[Corollary 2.8]{FLL}).
\end{proof}

\begin{remark}
For $s\geq 2$, it is useful to remember that a necessary and sufficient condition for the existence of an NSC matrix $A\in M_{s\times l}(R/\langle \ga \rangle)$ (or even $A\in M_{s\times l}(R)$) is that $s\leq l \leq |R/\langle \ga \rangle|$; see \cite[Proposition 3.3]{BN} (or \cite[Proposition 1]{As}).
\end{remark}

We now state the following known result over finite fields for a later use, although it is valid more generally over finite chain rings (as such rings are Frobenius, see \cite[Corollary 4.12]{FLL}).
\begin{prop}\label{min distance} $($\cite[Corollary 4.12]{FLL}$)$
Let $A\in M_{s\times s}(R/\langle \ga \rangle)$ be an NSC matrix and $\C_1, \dots, \C_s$ linear codes over $R/\langle \ga \rangle$ of the same length. Then we have the following bounds:
$$d([\C_1 \dots \C_s]\,A) \geq \min \{s\,d(\C_1), \dots, 2 d(\C_{s-1}), d(\C_s)\}$$
and
$$d(([\C_1 \dots \C_s]\,A)^\perp) \geq \min \{d(\C_1^\perp), 2\, d(\C_2^\perp), \dots, s\, d(\C_s^\perp)\}.$$
Furthermore, if $\C_s \subseteq \C_{s-1} \subseteq \dots \subseteq \C_1$, then all the above inequalities are equalities.
\end{prop}

\begin{dfn} (\cite{NS}) Let $\C$ be a linear code over $R$. For $0\leq i \leq e-1$, we call the linear code $\ol{(\C:\ga^i)}$ the {\it $i$-torsion code associated to $\C$}, and we denote it by $T_i(\C)$, where $(C:\ga^i):=\{x\in R^m\,|\, \ga^ix\in C\}$
is {\it the submodule quotient code of $C$ by $\ga^i$}.
\end{dfn}

\begin{remark}\label{nested} For a linear code $\C$ over $R$, the following nesting property is obvious:
$$T_0(\C)\subseteq T_1(\C) \subseteq \dots \subseteq T_{e-1}(\C).$$
\end{remark}

\begin{lem} \label{Tor2} Let $\C$ be a linear code over $R$ and $0\leq i \leq e-1$. Then,
\begin{itemize}
\item[1.] $T_i(\C^\perp)=(T_{e-1-i}(\C))^\perp$,
\item[2.] $H(T_i(\C)) \subseteq T_{e-1}(H(\C))$, and
\item[3.] if $\C$ is LCD, then so is $T_i(\C)$.
\end{itemize}
\end{lem}

\begin{proof}\hfill
\begin{itemize}
\item[1.] See \cite{NS}.
\item[2.] Let $\ol{z}\in H(T_i(\C))$. Then, by part 1, $\ol{z}\in T_i(\C)\cap T_{e-1-i}(\C^\perp)$. Then for any $z'\in (\C:\ga^i)\cap (\C^\perp: \ga^{e-1-i})$ with $\ol{z}=\ol{z'}$, we have $\ga^i z' \in \C$ and $\ga^{e-1-i} z' \in \C^\perp$. By the linearity of $\C$ and $\C^\perp$, we get $\ga^{e-1}z' \in \C$ and $\ga^{e-1}z' \in \C^\perp$. So, $z'\in (H(\C):\ga^{e-1})$ and thus $\ol{z'}=\ol{z}\in T_{e-1}(H(\C))$.
\item[3.] See a different proof of this fact in \cite[Theorem 3.4]{LL1}. Suppose that $\C$ is LCD; so $H(\C)=\{0\}$. By part 2, $$H(T_i(\C))\subseteq T_{e-1}(H(\C))=\ol{(H(C):\ga^{e-1})}=\ol{(\{0\}:\ga^{e-1})}.$$
    Let $y=(y_1, \dots, y_m)\in (\{0\}:\ga^{e-1})$. We claim that $y_i \in \langle \ga \rangle$ for every $i=1, \dots, m$. If $y_j=0$ for some $1\leq j \leq m$, then obviously $y_j\in \langle \ga \rangle$. Assume that $y_j \neq 0$ for some $1\leq j \leq m$. As $\ga^{e-1}y=(\ga^{e-1}y_1, \dots, \ga^{e-1}y_m)=(0, \dots, 0)$, it follows in particular that $\ga^{e-1}y_j=0$. Since $\ga^{e-1}$ and $y_j$ are both nonzero, $y_j$ is a zero divisor. Since $R$ is local with its maximal ideal $\langle \ga \rangle$, all non-units of $R$ are elements of $\langle \ga \rangle$. Since $R$ is finite, the nonzero non-units of $R$ are precisely the zero divisors of $R$; that is, $\langle \ga \rangle$ consists exactly of the zero element and the zero divisors of $R$. Thus, $y_j\in \langle \ga \rangle$ in this case as well. Now, since $y_i\in \langle \ga \rangle$ for every $i=1, \dots, m$, $\ol{y_i}=\ol{0}$ in $R/\langle \ga \rangle$ for every $i=1, \dots, m$ and hence $\ol{y} = \ol{0}$ in $(R/\langle \ga \rangle)^m$. This shows that $\ol{(\{0\}:\ga^{e-1})}=\{\ol{0}\}$ and, therefore, $H(T_i(\C))=\{\ol{0}\}$. Hence, $T_i(\C)$ is LCD.
\end{itemize}
\end{proof}

\begin{remark} Notice that the converse of part 3 of the above result is true in general (see \cite[Example 1]{LL1}).
\end{remark}

The following result gives a way of constructing LCD codes over the field $R/\langle \ga \rangle$ of certain parameters given an LCD code over the ring $R$.
\begin{thm}\label{torsion theorem}
Let $\C$ be an LCD code of length $m$ over $R$ and $A\in M_{s\times s}(R/\langle \ga \rangle)$ non-singular.
\begin{itemize}
\item[1.] If $AA^t$ is diagonal, then $[T_{i_1}(\C) \dots T_{i_s}(\C)]\,A$ is an LCD code of length $ms$ over $R/\langle \ga \rangle$ for $0\leq i_j \leq e-1$ and $1\leq j \leq s$. Moreover, if $A$ is NSC and $i_s \leq i_{s-1} \leq \dots \leq i_1$, then
    $$d([T_{i_1}(\C) \dots T_{i_s}(\C)]\,A) =\min\{sd(T_{i_1}(\C)), \dots, 2d(T_{i_{s-1}}(\C)), d(T_{i_s}(\C))\}$$ and
    $$d(([T_{i_1}(\C) \dots T_{i_s}(\C)]\,A)^\perp)=\min\{d((T_{i_1}(\C))^\perp), 2 d((T_{i_2}(\C))^\perp), \dots, sd((T_{i_s}(\C))^\perp)\}.$$
\item[2.] If $AA^t$ is antidiagonal, then $[T_{i_1}(\C) T_{i_2}(\C) \dots T_{i_2}(\C) T_{i_1}(\C)]\,A$ is an LCD code of length $ms$ over $R/\langle \ga \rangle$ for $0\leq i_j \leq e-1$ and $1\leq j \leq \lfloor \frac{s+1}{2} \rfloor$.
\item[3.] $[\underbrace{T_i(\C) \dots T_i(\C)}_s]\,A$ is an LCD code of length $ms$ over $R/\langle \ga \rangle$ for every $0\leq i \leq e-1$. Moreover, if $A$ is NSC, then
    $$d([\underbrace{T_i(\C) \dots T_i(\C)}_s]\,A)=d(T_i(\C))$$ and
    $$d(([\underbrace{T_i(\C) \dots T_i(\C)}_s]\,A)^\perp)= d((T_i(\C))^\perp).$$
\item[4.] If $A$ is upper triangular (resp. lower triangular), then $[T_{i_1}(\C) \dots T_{i_s}(\C)]\,A$ (resp. \linebreak $[T_{i_s}(\C) \dots T_{i_1}(\C)]\,A$) is an LCD code of length $ms$ over $R/\langle \ga \rangle$ for $0\leq i_j \leq e-1$, $i_j \leq i_{j+1}$, and $1\leq j \leq s-1$.
\end{itemize}
\end{thm}
\begin{proof}
To begin with, it follows from Lemma \ref{Tor2} that $T_i(\C)$ is LCD for every $0\leq i \leq e-1$. Now, the rest of the proof follows from Corollary \ref{LCD cor 1}, Proposition \ref{min distance}, and Remark \ref{nested}.
\end{proof}

\begin{remark}\hfill
\begin{enumerate}
\item[1.] In part 2 of Theorem \ref{torsion theorem}, we did not bother to use Proposition \ref{min distance} to declare a sharp estimate on the minimum distance of the matrix-product code. The reason is that we would then need the assumption $i_1 \leq i_2 \leq \dots \leq i_{\lfloor \frac{s+1}{2}\rfloor} \leq \dots \leq i_2 \leq i_1$, which would imply the equality of all of the input codes. But, if so, it would be better to use part 3 as the only requirement in part 3 is that $A$ be non-singular, which is more general than the requirement in part 2 that $A$ be non-singular and $AA^t$ anti-diagonal.
\item[2.] In part 4 of Theorem \ref{torsion theorem}, it should be obvious that an upper triangular (or lower triangular) matrix is never an NSC matrix if $s\geq 2$. This is why we did not consider using Proposition \ref{min distance} here.
\end{enumerate}
\end{remark}

\begin{example}
Consider $\Z_4/\langle 2 \rangle=\F_2$. Let $\C$ be the $[8,4,2]$-code over $\Z_4$ generated by the matrix
$$G=\left(\begin{array}{cccccccc} 1&0&0&0&0&1&2&1\\ 0&1&0&0&1&2&3&1\\ 0&0&1&0&0&0&3&2\\ 0&0&0&1&2&3&1&1 \end{array}\right),$$
which is in the standard form $(I_4\; A_{0,1})$ (see \cite{NS}). By \cite[Theorem 3.5]{LL1}, $\C$ is LCD over $\Z_4$ since $\mbox{det}(GG^t)\in U(\Z_4)$. By Lemma \ref{Tor2}, $T_0(\C)=\ol{(\C:2^0)}=\ol{\C}$ and $T_1(\C)=\ol{(\C:2)}$ are LCD binary codes. We can see, by \cite{NS} again, that $T_0(\C)$ and $T_1(\C)$ are equal as they have the same generating matrix over $\F_2$
$$\ol{G}=\left(\begin{array}{cccccccc} 1&0&0&0&0&1&0&1\\ 0&1&0&0&1&0&1&1\\ 0&0&1&0&0&0&1&0\\ 0&0&0&1&0&1&1&1 \end{array}\right).$$
Set $T(\C)=T_0(\C)=T_1(\C)$. Then $T(\C)$ is an $[8,4,2]$ binary code. Now, by Theorem \ref{torsion theorem}, for any non-singular matrix $A\in M_{2\times 2}(\F_2)$, the binary matrix-product code $[T(\C), T(\C)]\, A$ is an LCD code with parameters $[16,8,2]$ (see \cite[Theorem 1]{HLR}).
\end{example}

\section*{Acknowledgement}
A. Deajim and M. Bouye would like to express their gratitude to King Khalid University for providing administrative and technical support.

\end{document}